\documentclass[conference]{IEEEtran}
\IEEEoverridecommandlockouts

\usepackage[margin=0.73in]{geometry}

\usepackage{dblfloatfix}
\usepackage{cite}
\usepackage{multicol}
\usepackage{amsmath}
\usepackage{subcaption}
\usepackage{fancyhdr}
\fancypagestyle{firstpage}{%
  \lhead{This work has been accepted for publication in the IEEE WCNC 2024 Conference.}
}

\usepackage{amssymb}
\usepackage{amsfonts}
\usepackage{bbm}
\usepackage[T1]{fontenc}
\usepackage{algpseudocode}
\usepackage[utf8]{inputenc}
\usepackage[english]{babel}
\usepackage[usenames, dvipsnames]{color}
\usepackage{pifont}
\usepackage{booktabs}
\usepackage{enumitem}
\usepackage{mathtools}
\usepackage{ellipsis}
\usepackage{textcomp}
\usepackage{cite}
\usepackage[pdftex]{graphicx}
\usepackage{subcaption}
\usepackage{epsfig}
\usepackage{epstopdf}
\usepackage{caption}
\usepackage[ruled,vlined,commentsnumbered]{algorithm2e}
\usepackage{fixltx2e}
\usepackage{algpseudocode}
\usepackage{epstopdf}
\usepackage{psfrag}
\usepackage{grffile}
\usepackage{amsthm}
\SetKwInput{KwIn}{Inputs} 

\usepackage{caption}

\newtheorem{thm}{\textbf{Theorem}}

\theoremstyle{definition}

\theoremstyle{remark}

\theoremstyle{plain}

\newcommand{\RNum}[1]{\uppercase\expandafter{\romannumeral #1\relax}}
\usepackage{setspace}
\setlist[itemize]{leftmargin=*}

\def\BibTeX{{\rm B\kern-.05em{\sc i\kern-.025em b}\kern-.08em
    T\kern-.1667em\lower.7ex\hbox{E}\kern-.125emX}}
\begin{document}

\title{Modeling and Performance Analysis of CSMA-Based JCAS Networks
} 
\author{\IEEEauthorblockN{ Navid Keshtiarast, Pradyumna Kumar Bishoyi, and Marina Petrova}
\IEEEauthorblockA{Mobile Communications and Computing, RWTH Aachen University, Aachen, Germany \\
Email: \{navid.keshtiarast, pradyumna.bishoyi, petrova\}@mcc.rwth-aachen.de}
 \vspace{-8mm}
}
\maketitle
\thispagestyle{firstpage}
\begin{abstract}
  Joint communication and sensing (JCAS) networks are envisioned as a key enabler for a variety of applications which demand reliable wireless connectivity along with accurate and robust sensing capability. When sensing and communication share the same spectrum, the communication links in the JCAS networks experience interference from both sensing and communication signals. Therefore, it is crucial to analyze the interference caused by the uncoordinated transmission of either sensing or communication signals, so that effective interference mitigation techniques could be put in place. We consider a JCAS network consisting of dual-functional nodes operating in radar and communication modes. To gain access to the shared communication channel, each node follows carrier sense multiple access (CSMA)-based protocol. For this setting, we study the radar and communication performances defined in terms of maximum unambiguous range and aggregated network throughput, respectively. Leveraging on the stochastic geometry approach, we model the interference of the network and derive a closed-form expression for both radar and communication performance metrics. Finally, we verify our analytical results through extensive simulation.
\end{abstract}


\section{Introduction}
\label{sec:Introduction}
With the increasing demand for high-precision localization and sensing-based services, traditional wireless networks are undergoing a major paradigm shift from predominantly communication-oriented networks to the ones that provide ubiquitous connectivity and sensing in the form of joint communication and sensing (JCAS) networks \cite{liu2022,pin2021}. The integration of sensing functionalities could potentially boost the system performance and may assist for mutual benefits, which are not present in traditional communication systems \cite{Chiriyath2016}. 
However, to fully exploit this mutual performance gain, one must carefully analyze the goals and constraints of both communication and sensing functionalities. The main goal of sensing is to extract useful information from a noisy environment, whereas the main objective of communication is to effectively transmit and recover information in a noisy environment \cite{liu2020}. Given the diverse goals, understanding the performance trade-offs and quantifying mutual gains is crucial for co-designing both functionalities to improve the overall system performance.

\textit{Related works}: In the recent past, research efforts have been devoted to analyzing the performance of the JCAS networks under different system assumptions and parameters. In \cite{munari_pimrc}, the authors analyzed the coexistence performance of uncoordinated radars and communication devices using stochastic geometry to explore the impact of network density, packet length, and antenna directivity on radar detection and communication network throughput. 
In \cite{ping_letter}, the coexistence of JCAS nodes operating in 60 GHz frequency is investigated for different densities and antenna directivity. This work was extended in \cite{ping_ICNC}, where the performance of radar and communication is analyzed by studying radar range for different radar update rates and communication network throughput for various network densities. In \cite{Ram2022}, the authors deliberate a generalized analytical framework to get the highest throughput by changing the radar duty cycle and pulse repetition interval. The JCAS in the context of the cellular network is analyzed in \cite{jeff2022}. By using the stochastic geometry approach, the authors derive the upper and lower bounds for sensing and communication coverage probability and ergodic capacity of the JCAS network. 

\textit{In most of the aforementioned works, the authors have considered JCAS-based networks with contention-free channel allocation for the communication and sensing  (i.e. cellular networks operating in the licensed band) or ALOHA random access protocol for communication, which suffers from low throughput. However, many current communication standards, including Wi-Fi and 5G new radio unlicensed (NR-U), use carrier sensing (CS)-based access protocol for channel access in the 5 GHz and 6 GHz bands.} JCAS systems emerging in these bands will have to share the spectrum using CS-based access schemes. Since both sensing and communication share the same spectrum, the sensing performance is coupled with the choice of communication medium access protocol. In addition, sensing echo signals suffer double path loss effect \cite{tait2005}; therefore, in the case of dense JCAS network scenarios, the sensing performance is going to be severely degraded due to interference from the communication signal.

In contrast to the previous works, we consider CS-based protocol for communication and study the performance of both sensing\footnote{In our work, we use the terms radar and sensing interchangeably. More specifically, we consider radar-based sensing, where the radio signal is used to detect the presence of the target and the maximum unambiguous range is used to evaluate the sensing performance.} (radar) and communication functionalities in a JCAS network operating in sub-7 GHz. Since the unlicensed sub-7 GHz band is already home to diverse technologies and emerging cutting-edge applications that necessitate simultaneous sensing of the surroundings and communication among multiple co-located users (e.g. XR gaming, Metaverse applications), it is essential to conduct interference analysis and understand how this affects both the communication and sensing performance. 
The key questions we try to address here are:
\begin{itemize}
    \item \textit{How to obtain maximum detectable range while considering the desirable false alarm rate and medium access probability of communication nodes?} 
    \item \textit{What is the impact of radar transmission on the performance of CS-based communication?}
\end{itemize}

In this paper, we analyze the performance of a JCAS network in which each node switches between radar and communication mode alternatively in an uncoordinated manner. In communication mode, each node follows a CS-based medium access protocol to access the shared channel. We model the locations of the nodes as a Poisson point process (PPP) and derive a closed-form expression for both radar and communication performance metrics. 
To this end, the main contributions of this work are summarized as follows:
\begin{itemize}
\item First, we develop an analytical framework to evaluate the performance of CSMA-based JCAS networks. Specifically, based on the stochastic geometry approach, we model the interference caused by the uncoordinated transmission of either radar pulse or data transmission in a large-scale JCAS network.
\item Thereafter, we analyze the sensing performance in terms of maximum unambiguous range in the presence of constant false alarm rate and communication transmission probability. Moreover, we derive the closed-form expression for the medium access probability and aggregated network throughput density to characterize the communication performance.
\item Finally, the accuracy of the proposed analytical model is validated through the extensive simulation results.
\end{itemize}

This paper is organized as follows: A discussion of the system model and the performance metrics follows in Section II. Section III shows and discusses our analytical and simulation results. In Section IV, we conclude the paper.
\section{System Model}
We consider a JCAS network consisting of dual-functional nodes located according to a homogeneous PPP $\Phi = \{\mathbf{x}_i\} $ with intensity $\lambda$. As shown in Figure \ref{fig:System_model1}\footnote{Some icons taken from https://www.flaticon.com.}, each node supports both sensing and communication functions and acts as a radar and communication node. In radar mode, each node functions as a monostatic radar, acting as both a radar transmitter and receiver. During this mode, the node senses its surrounding by transmitting a radar pulse and waiting for an echo signal. Following the radar mode, the node enters communication mode. We assume that each node in the network has a dedicated receiver and that each node uses a CSMA/CA-based medium access control protocol to compete for the channel. Each node employs its radar and communication mode in a time division manner, as shown in Figure \ref{fig:System_model}. We assume that the time is slotted to facilitate analytical tractability. Let $M=M_r+M_c$ represent the total time slots, where $M_r$ and $M_c$ represent the number of slots used for radar sensing and data transmission purposes, respectively. Each node first starts in radar mode by sending a pulse in one-time slot and waiting for echoes in $M_r-1$ slots. The node immediately switches to communication mode for the remainder of the $M_c$ slot duration. In this mode, the node performs carrier sensing and transmits only when the channel is idle. This cycle is repeated for each node in the network. Clearly, there exists a trade-off in the allocation of time slots for radar and communication modes. The communication performance will suffer as more slots are reserved for the radar mode. To capture this trade-off, we define the parameter $\epsilon=M_r/M$, which represents the fraction of slots dedicated to the radar mode.
\vspace{-0.1cm}
\begin{figure}[t!]
	\centering
	\includegraphics[width=0.44\textwidth,trim = 20mm 80mm 20mm 72mm,clip]{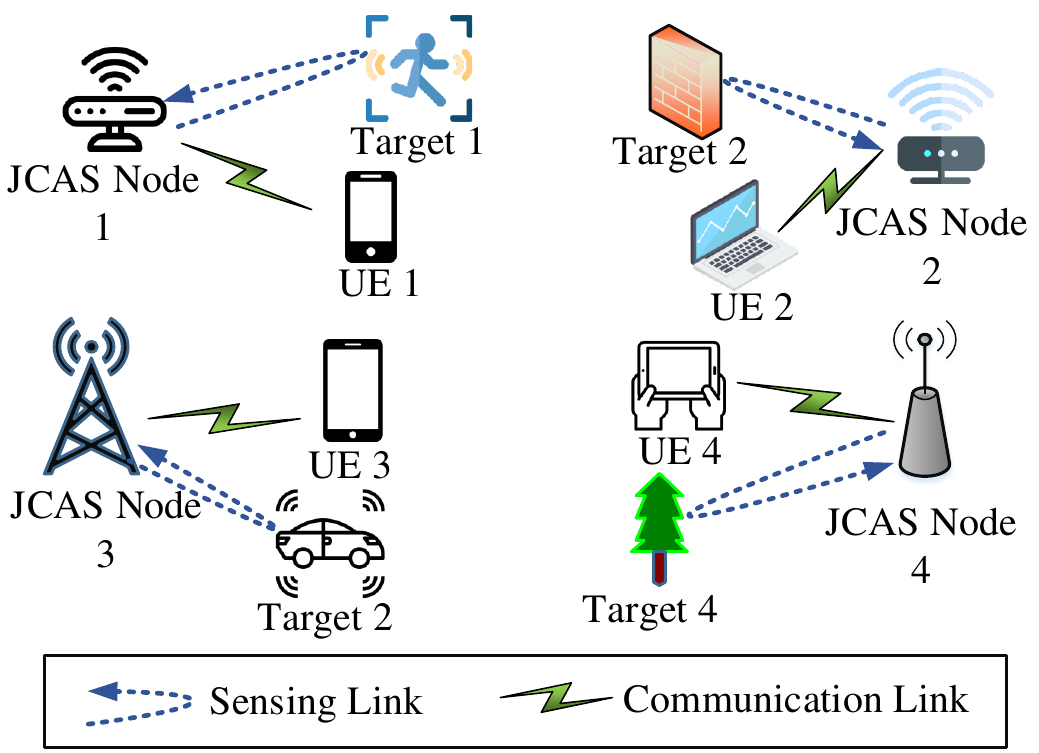}
	\caption{System diagram of JCAS network}
	\label{fig:System_model1}
	\vspace{-0.6cm}
\end{figure}

At any random snapshot of time, one node is either in radar mode or communication mode. Let $\tau$ be the fraction of nodes in radar mode and $(1-\tau)$ be the fraction in communication mode. The set of nodes in radar and communication mode can be realized from independent thinning of original PPP $\Phi$ \cite{haenggi2012stochastic}. By the corresponding property of PPP, $\Phi_r$ is a PPP with intensity $\lambda_r = \tau\lambda$ for the number of nodes in radar mode and  $\Phi_c$ with intensity $\lambda_c = (1-\tau)\lambda$ for communication nodes. The overall network performance is dependent on the value of $\tau$, and we are interested in analyzing radar and communication performance for a given $\tau$. In the following subsections, we explain in detail the radar and communication performance.
\vspace{-0.3cm}
\begin{figure}[h!]
	\centering
	\includegraphics[width=0.5\textwidth,trim = 40mm 102mm 60mm 16mm,clip]{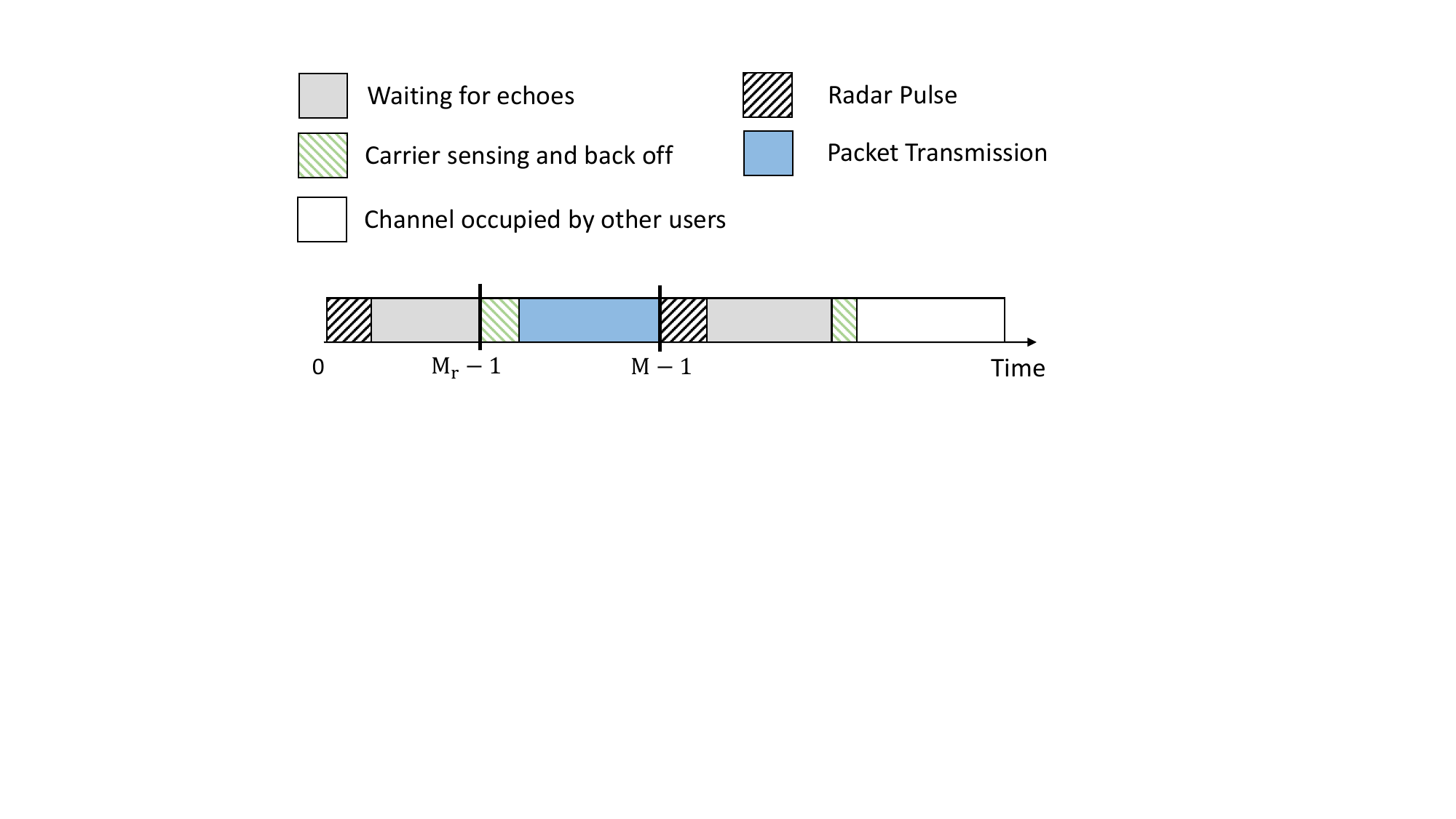}
	\caption{Time diagram of JCAS node over two cycles }
	\label{fig:System_model}
\end{figure}
\vspace{-5mm}
\subsection{Radar Performance}
To evaluate radar performance, we consider the maximum unambiguous range as one of the key performance metrics, which is defined as the maximum range at which a target can be detected reliably. In our scenario, each node in radar mode acts as monostatic radar, and target detection decisions are taken based on the receiving echo power. In classical radar theory \cite{tait2005}, reliable target detection depends on two parameters. First, the probability of detection, $P_D$ is defined as the probability that the echo power is above a certain predefined threshold. Second, the probability of false alarm, $P_{fa}$, which is defined as the interference received at the radar receiver is above a certain threshold. Here, we try to analyze the effect of interference on the maximum unambiguous range of the radar.  
Let $P_{tx}$ be the transmitting power of the node. Then, the expression for received echo power is
\begin{equation}
    P_r = \frac{P_{tx}G^2c^2\sigma_{rcs}}{(4\pi)^3f_c^2d^{4}}
    \vspace{-2mm}
\end{equation}
where $G = 4\pi/\phi^2$ is the antenna gain over the main beam of width $\phi = [0,2\pi]$, $\sigma_{rcs}$ is the radar cross section (RCS) of the target, $c$ is the speed of the light, $f_c$ is the operating frequency, and $d$ is the distance between the radar node and the target. 

In radar mode, each node waits for $M_r-1$ slots for the echo. Since the radar node is surrounded by other radar and communication nodes, interference is inevitable. For this setup, the aggregated interference at the node at slot $m \in M_r-1$ is 
\vspace{-2mm}
\begin{equation}
    I_m = \sum_{i \in \Phi}\kappa P_{tx}G^2 d_i^{-4}
    \vspace{-2mm}
\end{equation}
where $\kappa = \frac{c^2}{(4\pi)^3f_c^2}$ and $d_i$ is the interferer-receiver distance. For the successful detection of a target, the received echo power must be above a certain threshold $\theta$, i.e. $P_D = P(P_r+I_m >\theta)$. Also, in converse, the false alarm is triggered when the interference value is above the threshold $\theta$, i.e. $P_{fa} = P(\bigcup_{m=1}^{M_r-1}I_m >\theta)$.

To evaluate the interference, the probability of collision during the echo waiting duration, i.e. $M_r-1$ slots, has to be calculated. Since considering the joint interference distribution is mathematically difficult, we use the nearest interferer approximation \cite{baccelli2010stochastic}, i.e. the closest neighbour in our analysis. Further, we assume that when the node under investigation is inside the main beam of its neighbouring radar node, it will receive maximum interference. Since the antenna pattern is randomly oriented, the probability of the above event is $\frac{\phi}{2\pi}$ and the intensity of nodes in radar mode is $\lambda_r^{'} = (\phi/2\pi)\lambda_r$. Based on these assumptions, we derive the expression for the maximum unambiguous range of the radar.
\vspace{-0.15cm}

\begin{thm} Given the desired false alarm probability ($P_{fa}$), the maximum unambiguous radar range is
\begin{equation} \label{eqn9}
\begin{aligned}
  d_{rm}=\bigg(\frac{\sigma_{rcs}f_c^2}{4\pi^3c^2}\bigg)^{1/8}\bigg(-\ln{\bigg(1-\frac{P_{fa}}{C(M_r,\epsilon,q_w)}}\bigg)\bigg)^{1/4}
\end{aligned}
\end{equation}
where $ C(M_r,\epsilon,q_w)= 1-\frac{1}{M}-\sum_{i=M_r}^{M-1}{\frac{1}{M}(1-q_w)^{N_i}},$,
$N_i = \min(M_r-1,M_r/\epsilon-i)$, and $q_w$ is the transmitting probability of the nodes in communication mode. 
\end{thm}
\begin{proof}:
The proof is given in one of our previous works \cite[pp.2]{ping_letter}. Due to space limitation of the paper, we have omitted it here.
\end{proof}
\vspace{-3mm}
Clearly, the radar performance is coupled with the transmitting probability ($q_w$) of communication nodes. In the subsequent subsection, we derive the expression for $q_w$ and analyze the communication network throughput. 

\subsection{Communication Performance}
In this subsection, we study the throughput of the nodes operating in communication mode. To access the shared communication medium between multiple JCAS nodes, a CSMA/CA-based channel access mechanism is used. The communication between transmitter and receiver is successful if the transmitter gets access to the channel and the signal-to-interference-plus-noise ratio (SINR) at the receiver is above a certain decoding threshold. Since the nodes are surrounded by both radar and communication nodes, the aggregated interference on the communication performance is non-negligible. The goal of this subsection is to quantify this interference and model the probability of successful reception. 

As explained above, each communication node transmits its data packet to a dedicated receiver located at $r_c$ distance apart. The position of the receiver is assumed to be uniformly distributed on a circle of radius $r_c$ centered around its transmitter node. For example, if a node $i$ has a transmitter located at $x_i\in \Phi_c$, then the position of its receiver is $y_i = x_i+r_cz(\theta_i)$, where $z(\theta)=(\cos(\theta),\sin(\theta))$ and $\theta$ is a uniformly distributed from $0$ to $2\pi$. Since each JCAS transmitter node communicates with the nearest receiver, no other JCAS node can be closer than $r_c$. In that case, the probability density function of $r_c$ is
\vspace{-1mm}
\begin{equation} \label{dist_pdf}
  f_c(r_c)=\lambda_c 2\pi r_c \exp(-\lambda_c \pi r_c)
   \vspace{-1mm}
\end{equation}
For the sake of tractability, we consider only the downlink scenario in communication mode.

\subsubsection{Channel model}
The transmit power of each JCAS node is assumed to be $P_{tx}$. We denote $l(d)$ as the path loss of the link between transmitter and receiver with distance $d$. The path loss model considered for our analysis is given by $l(d) = 20\log_{10}(\frac{4\pi c}{f_c}) + 10\alpha\log_{10}(d) $, where $c$ is the speed of light, $f_c$ denotes the operating frequency, and the path loss exponent $\alpha=3$. The small-scale fading between communication node transmitter $i$ and receiver is $h_{ii}^c$ and the received power at the receiver located at ${x}_i \in \Phi_c$ is $P_{rx}^0(x_i) = P_{tx}h_{ii}^c/l(\lVert x_i\rVert) $. Further, the small-scale fading between receiver $i$ and interfering communication node $j$ is denoted as $h_{ij}^c$ and between receiver $i$ and interfering radar transmitter $j$ is denoted as $h_{ij}^r$. Therefore, the interference power at the receiver $i$ from the communication node located at ${x}_j \in \Phi_c$ is $P_{rx}^c(x_j) = P_{tx}h_{ij}^c/l(\lVert x_j-x_i\rVert) $ and the interference power from the radar node located at $y \in \Phi_r$ is $P_{rx}^r(y) = P_{tx}h_{ij}^r/l(\lVert y-x_i\rVert) $. We assume that all channels experience Rayleigh fading with mean $\mu = 1$.

\subsubsection{Channel access mechanism}
Each JCAS node contends the shared channel using CSMA/CA-based protocol. In the case of CSMA/CA, each node senses the channel before the transmission. The channel is considered as \textit{busy} if it detects any other transmission within its sensing range. In our considered scenario, a JCAS node $i$ contending for the channel is surrounded by the JCAS nodes, which are either in radar mode or communication mode. Therefore, the node $i$ will consider the channel busy if any radar pulse or other communication signal that exceeds its carrier sensing threshold is detected. As soon as the node $i$ observes the channel \textit{idle}, it initiates its back-off timer, which defines the random amount of time slot the node waits before starting its transmission. We model the back-off timer of each node $i$ as an independent mark $t_i$, which is uniformly distributed in the interval $[0,1]$. The node gets access to the channel when it chooses the smaller $t_i$ value \cite{baccelli2010stochastic,jeff2016}.

Let $e_i$ denote the medium access indicator of node $i$ located at $x_i \in \Phi_c$, i.e. whether the node gets access to the channel or not. The value of $e_i$ is $1$ if the node $i$ is allowed to transmit and $0$ otherwise. Similar to the radar analysis, we expect that communication nodes will be most impacted by neighbouring radar nodes if and only if they are within the radar node's main lobe. Therefore, the set of interfering radar nodes ($\Phi_r^{'}$) is a thinning version on ($\Phi_r$) with intensity $\lambda_r^{'} = \lambda_r(\phi/2\pi)$. Further, following the analysis in \cite{baccelli2010stochastic}, we derive the expression for medium access indicator $e_i$ of node $i$ as,
\vspace{-0.24cm}
\begin{multline}
\label{eqn1}
e_i = \prod_{y \in \Phi_r } \mathbbm{1}_{P_{rx}^r(y)\leq P_{th}}\prod_{x_j \in \Phi_c^{'} } \big({\mathbbm{1}_{t_j^w\geq t_i^w}}\\+\mathbbm{1}_{t_j^w\leq t_i^w}\mathbbm{1}_{P_{rx}^c(x_j)\leq P_{th}}\big)
\end{multline}
where $\Phi_c^{'} = \Phi_c \setminus \{x_i\}$, $P_{rx}^r$ is the receiving power from the radar node, $P_{rx}^r$ is the receiving power from the communication node, and $P_{th}$ is the detection threshold for radar and communication signal. The first term signifies that the node will not transmit if it detects a radar pulse. The first part of the second term corresponds to the case when the back-off timer of $i$ is less than other nearby communication nodes, and it will transmit. The second part of the second term corresponds to the scenario where the back-off timer of $i$ is greater than the other node $j$, but the signal received from the node $j$ is lesser than the threshold $P_{th}$.

We are interested to determine the medium access probability (MAP) of node $i$, i.e. $\mathbf{P}(e_i=1)$. The MAP of node $i$ is dependent on its location with respect to other nodes (both radar and communication) and its back-off timer value. Further, the radar nodes operate in duty cycle pattern, i.e., it transmits its radar pulse for a faction $\eta$ of the time slot ($0<\eta\leq 1$) and waits for its echo for the rest of the duration. Since at a given time each radar node transmits independently with probability $\eta$, the radar nodes affecting the communication performance form a PPP with intensity $\eta\lambda_r$. Let the receiver of node $i$ is located at the origin, then the position of the transmitter is $x_i = (r_c,0)$. The expression for MAP, $q_w$, is, 
\begin{equation}
	\begin{aligned}
	&q_w =\mathbf{P}_0[e_i=1]\\
	&= \int_{0}^{\infty} \mathbb{E}[e_i|x_i=(r_c,0)]f_c(r_c) \,dr_c \\
	& \begin{multlined}\overset{\mathrm{(a)}}{=} \int_{0}^{\infty} \mathbb{E}\bigg[ \prod_{y \in \Phi_r } \mathbbm{1}_{P_{rx}^r(y)\leq P_{th}}\bigg]\mathbb{E}\bigg[\prod_{x_j \in \Phi_c^{'} } \big({\mathbbm{1}_{t_j^w\geq t_i^w}}\\+\mathbbm{1}_{t_j^w\leq t_i^w}\mathbbm{1}_{P_{rx}^c(x_j)\leq P_{th}}\big)\bigg]f_c(r_c) \,dr_c
	\end{multlined}\\
	&\overset{\mathrm{(b)}}{=} \int_{0}^{\infty}{\frac{1-\mbox{exp}(-K_c(r_c))}{K_c(r_c)}\mbox{exp}(-K_r)f_c(r_c) \,dr_c}
	\end{aligned}
 \vspace{-1mm}
\end{equation}
where $\mathbf{P}_0(\cdot)$ represents the Palm probability. Further, $K_r$ represents the expected number of radar nodes in $\mathbb{R}^2\setminus B(0,r_c)$ whose transmitted radar signal is above the detection threshold, where $B(0,r_c)$ is the closed ball with center $0$  and radius $r_c$, i.e., $K_r = \eta\lambda_r\int_{\mathbb{R}^2\setminus B(0,r_c)}\mbox{exp}(-\mu \frac{P_{th}}{P_{tx}}l(\lVert x \rVert) \,dx)$. Similarly, $K_c$ is the expected number of communication nodes in $\mathbb{R}^2\setminus B(0,r_c)$ region, i.e., $K_c(r_c) = \lambda_c\int_{\mathbb{R}^2\setminus B(0,r_c)}\mbox{exp}(-\mu \frac{P_{th}}{P_{tx}}l(\lVert x-r_c \rVert) \,dx)$. The (a) is obtained from the fact that $\Phi_r$ and $\Phi_c$ are independent PPPs, and (b) follows Slyvinak's theorem and probability generating functional (PGFL) of the PPP \cite{haenggi2012stochastic}.
\subsubsection{Probability of successful reception} 
The probability of successful reception ($P_s$) is defined as the instantaneous received SINR at the receiver of node $i$ is higher than the certain threshold $T$, given the medium access indicator ($e_i$) of its transmitter node is 1. Mathematically,
\begin{equation} \label{eqn6}
  P_s(T,\eta\lambda_r,\lambda_c)= \mathbf{P}(SINR_i \geq T | e_i = 1) 
\end{equation}
The received SINR of node $i$ is
\begin{equation}
    SINR_i = \frac{P_{rx}^0(x_i)}{\Gamma_i^r(y_j)+\Gamma_i^c(x_j)+\sigma_N^2}
\end{equation}
where $P_{rx}^0(\cdot)$ is the receiving power at receiver $i$ from its transmitter located at $x_i$. $\Gamma_i^r(\cdot)$ and $\Gamma_i^c(\cdot)$ denote the aggregated interference at node $i$'s receiver from surrounding radar and communication nodes, respectively. Finally, $\sigma_N^2$ denotes the noise power. The expressions for aggregated interference are
    \begin{equation}
	\begin{aligned}
	\Gamma_i^r(y) &= \sum_{y_j\in\Phi_r} P_{rx}^r(y_j) = \sum_{y_j\in\Phi_r}P_{tx}h_{ji}^r/l(\lVert y_j-x_i \rVert) \nonumber\\ 
	\Gamma_i^c(x_j) &= \sum_{x_j \in \Phi_c^{'}} P_{rx}^c(x_j) =\sum_{x_j \in \Phi_c^{'}}P_{tx}e_jh_{ji}^c/l(\lVert x_j-x_i \rVert) \nonumber
	\end{aligned}
	\end{equation}
where $e_j$ denotes the medium access indicator of node $j$. The final expression for $P_s$ for a given SINR threshold $T$ is
\begin{equation}
\begin{aligned}
    P_s &= \int_{0}^{\infty}\mbox{exp}\bigg(-\mu T l(r_c)\frac{\sigma_N^2}{P_{tx}}\bigg) \times\\  & \mbox{exp} \bigg(-\int_{\mathbb{R}^2}\frac{Tl(r_c)\eta\lambda_r}{l(\lVert x \rVert)+Tl(r_c)} \,dx\bigg) \times\\ & \mbox{exp}\bigg(\int_{\mathbb{R}^2\setminus B(0,r_c)}^{}{\frac{Tl(r_c)\lambda_w h_1(r_c,x)}{l(\lVert x \rVert)+Tl(r_c)}}\,dx\bigg)f_c(r_c) \,dr_c
\end{aligned}
\end{equation}
where $h_1(\cdot)$ denotes the probability that the other communication nodes transmit, given that the node $x_i$ transmits. The expression for $h_1(\cdot)$ is derived in \cite{jeff2016}. Notably, the probability of successful reception depends on the intensity of radar nodes ($\lambda_r$) which is a function of $\tau$, duty cycle of radar nodes ($\eta$), and the transmitting probability of other communication nodes ($h_1(\cdot)$).

Based on the fraction of time the JCAS node is in communication mode ($\epsilon$), the medium access probability ($q_w$), and the probability of successful reception ($P_s$), we can define the network throughput. The aggregated network throughput density is defined as the mean number of successful transmissions in a given area, which is given by:
\begin{equation} \label{eqn7}
\begin{aligned}
  \mathcal{T}=(1-\epsilon)\lambda_c q_w P_s
\end{aligned}
\end{equation}
\section{Performance Evaluation}
\subsection{Setup and Scenarios}
To verify the analytical model, we consider a JCAS network where the nodes are distributed following a homogeneous PPP. We consider two different simulation areas, $150\times150 m^2$ and $200\times200 m^2$, while keeping the network density the same. All the reported results are averaged over both the simulation areas. Further, we consider a random snapshot of time, where the nodes are either in communication or in radar mode. Each communication transmitter has one dedicated receiver, which is placed randomly at a distance of $20$ m from the transmitter. In the defined scenario, we vary the percentage of JCAS nodes working in radar mode ($\tau$), the fraction of slots dedicated to radar operation ($\epsilon$), and radar duty cycle ($\eta$) for different network densities. All simulation parameters are presented in Table \ref{sec:ScenariosandSetup}.
\begin{table}[h!] 
	\small
	\centering
	\caption{Simulation parameters}\label{sec:ScenariosandSetup}
	\renewcommand{\arraystretch}{1}
	{\begin{tabular}{|l|p{0.27\linewidth}|} \hline
			\textbf{Parameters} & \textbf{Values}\\ \hline 
			Simulation areas & $150 \times 150\mbox{ m}^2$ and $200 \times 200\mbox{ m}^2$\\ \hline
   			Network density, $\lambda$ (in nodes/${m^2}$) & $[10^{-6}- 1]$ \\ \hline
			Total no. of slots, $M$  & $ 100 $\\ \hline
			Transmit power, $P_{tx}$ & $ 23 $ dBm \\ \hline
			Mean value of radar cross-section, $\sigma_{rcs}$ & $ 10\mbox{ m}^2 $ \\ \hline
			Desired false alarm probability, $P_{fa}$ & $ 0.1 $\\ \hline
		    SINR threshold, $T$ & $ 5 $ dB\\ \hline
		    Detection threshold, $P_{th}$ & $ -62$ dBm\\ \hline
			Path loss exponent, $\alpha$ & $ 3 $ \\ \hline
			Operating frequency, $f_c$ & $6$ GHz\\ \hline
			Percentage of radar nodes, $\tau$ & $ 25\%, 50\%, 75\% $  \\ \hline
			Number of Seeds & $10^3$ \\ \hline
	\end{tabular}}
	\renewcommand{\arraystretch}{1}
	\vspace{-.4cm}
\end{table}

\subsection{Simulation Results and Discussion}
The communication and radar mode performances are analyzed based on average throughput and radar range, respectively. 
\begin{figure}[t!]
	\centering
	\includegraphics[width=0.45\textwidth,trim = 0mm 0.5mm 1mm 1mm,clip]{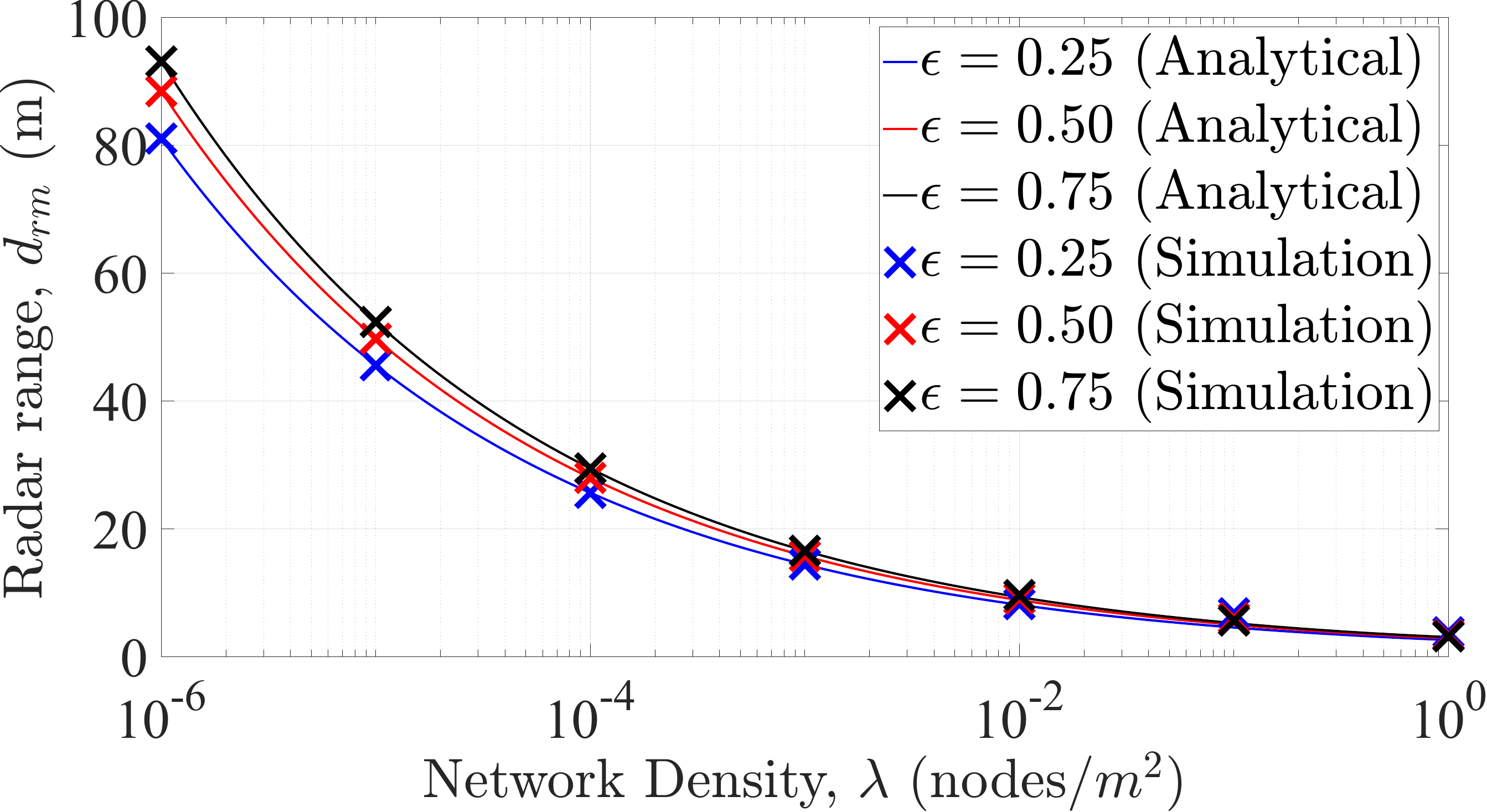}
	\caption{Radar range vs. network density for different $\epsilon$}
	\label{fig:radarrange}
	\vspace{-.45cm}
\end{figure}

Figure \ref{fig:radarrange} shows the radar range per total network node density $\lambda$ for different values of $\epsilon$. Here, $\tau=0.25$, whereas $\epsilon$ is changed to 0.25, 0.5, and 0.75. From the figure, it can be observed that as the number of nodes increases, the radar range decreases owing to increased interference from other radars. The reason for that is that all nodes use omni-directional antennas for radar operation. However, the radar range, even in dense network scenarios, is still suitable for low-range applications such as surveillance inside indoor residential areas. Further, it can be seen that the simulation results closely approximate to the analytical model. Moreover, the figure indicates that by dedicating more time to radar, i.e. $\epsilon$, the radar range increases. Lastly, it is worth mentioning that the radar range could be improved by either using directional antennas or by employing carrier sensing-based protocol for radars. 
\begin{figure}[h!]
    \vspace{-3mm}
	\centering
	\includegraphics[width=0.44\textwidth,trim = 0mm 1mm 0mm 0mm,clip]{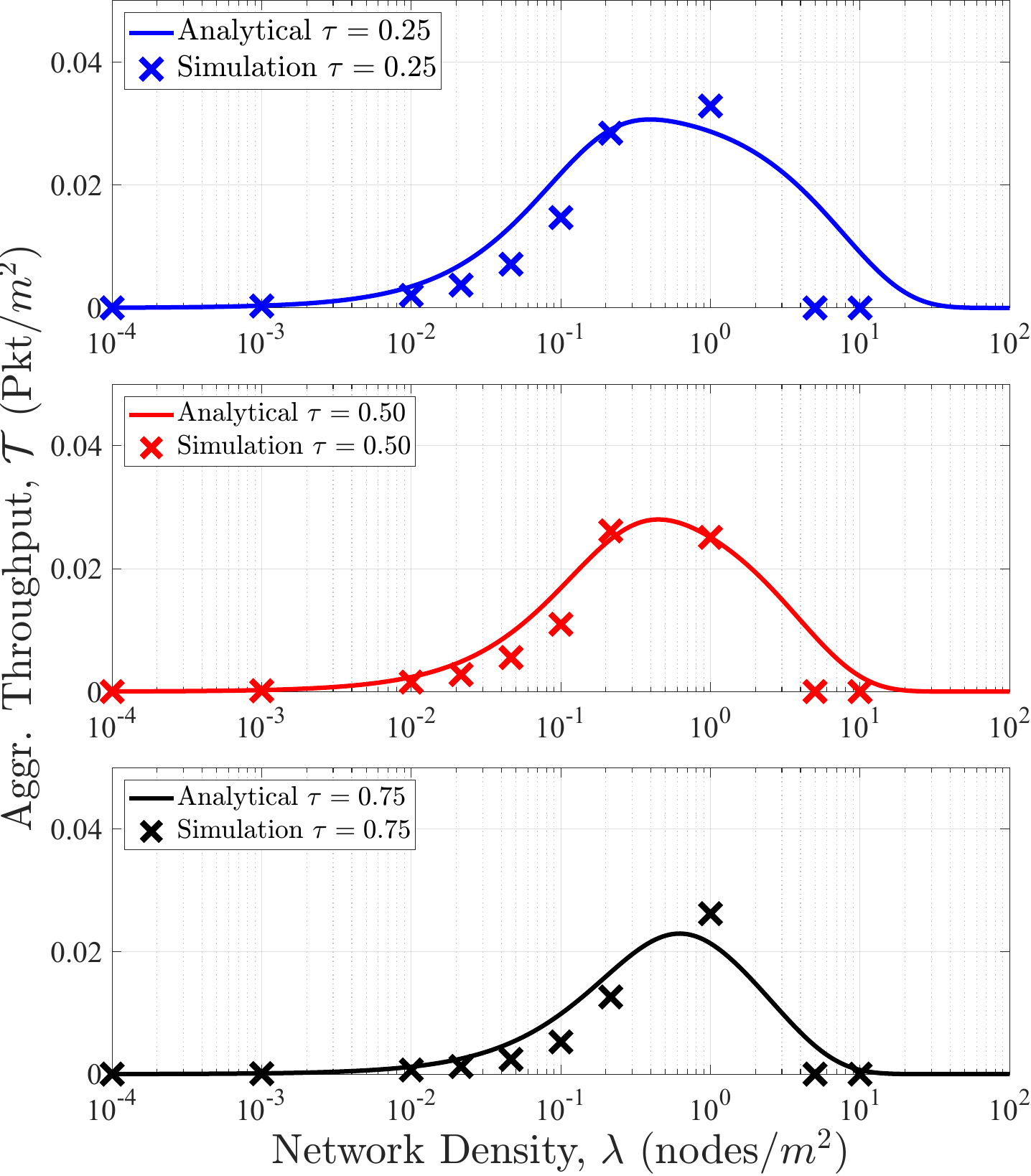}
	\caption{Throughput vs. network density for different $\tau$}
	\label{fig:Aggr_Throughput_tau}
	\vspace{-3mm}
\end{figure}

Figure \ref{fig:Aggr_Throughput_tau} depicts the aggregated network throughput per total network node density $\lambda$. Each subfigure represents the aggregated throughput for different $\tau$, i.e. percentage of nodes in radar mode. All curves indicate that the aggregate throughput of the network increases until it reaches a peak, followed by a reduction as network density increases. This is because with increasing the overall network density, the number of radar nodes increases, which results in higher radar-to-communication interference. In addition, we observe that the simulation results follow the same trend as the analytical one. The analytical results depict the upper bounds for the simulation results. Note that, for higher values of $\tau$, the discrepancy between the analytical and simulation results is mainly due to the impact of radar interference in a finite area (simulation) than in the case of an infinite area (stochastic geometry-based analytical model).
\begin{figure}[h!]
	\centering
	\includegraphics[width=0.44\textwidth,trim = 0mm 2.2mm 0mm 0.4mm,clip]{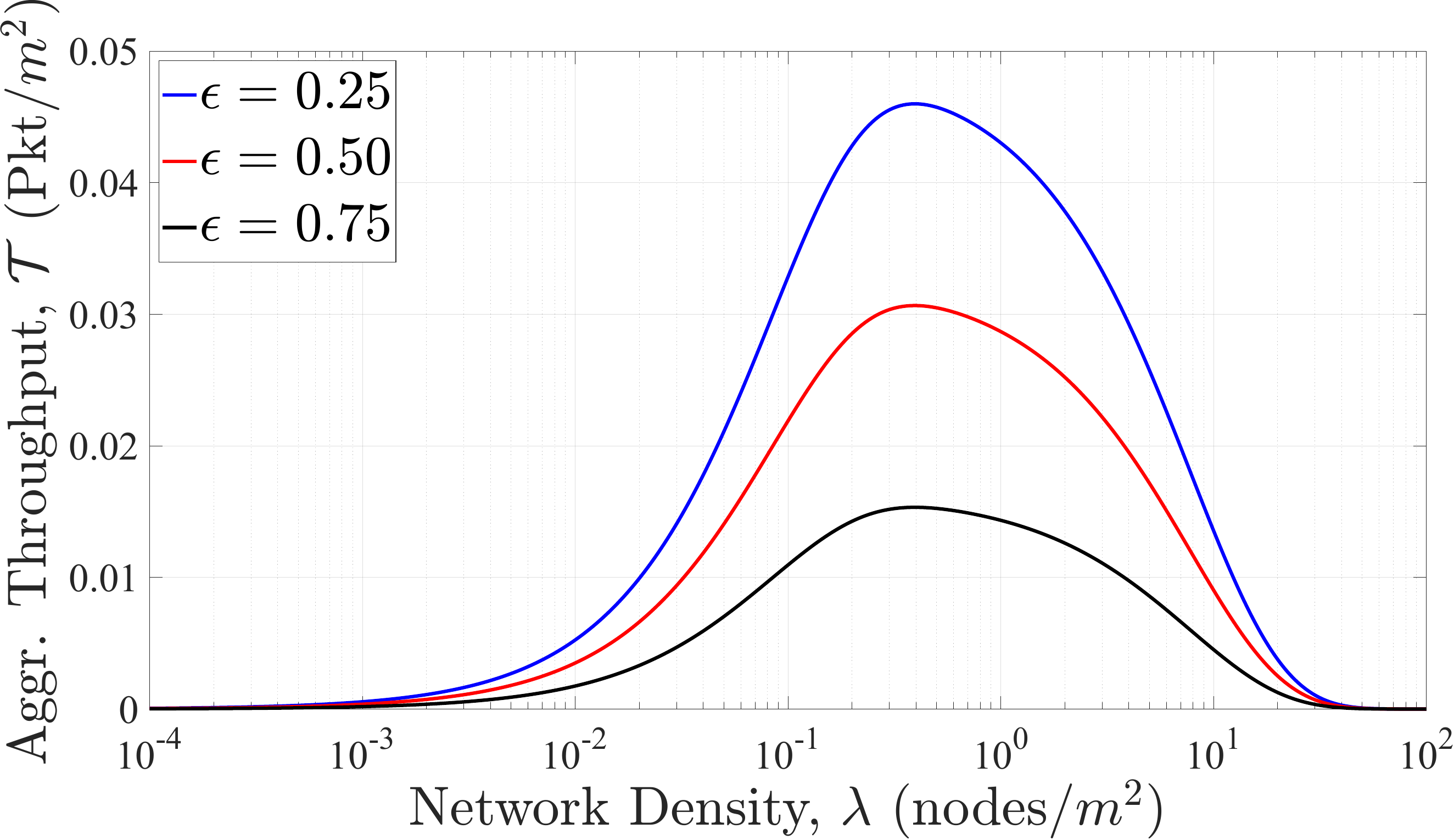}
	\caption{Throughput vs. network density for different $\epsilon$}
	\label{fig:Aggr_Throughput_epsilon}
	\vspace{-0.25cm}
\end{figure}

By analyzing the throughput curves in all three subfigures for different $\tau$, we observe that the throughput decreases by increasing the fraction of nodes in the radar modes. Since the nodes in the radar modes follow duty cycle transmission without carrier sensing, there will be increased interference by the radar nodes. According to the first subfigure, the maximum aggregated throughput can be reached at the node density of $0.5$, and the area under the throughput curves decreases with the increasing value of $\tau$. This is because, with the increase in the percentage of nodes in the radar mode ($\tau$), nodes in communication mode have to defer transmission. As a result, the throughput peak is achievable at the lower node density.

Figure \ref{fig:Aggr_Throughput_epsilon} shows the aggregated network throughput per network node density $\lambda$ for different values of $\epsilon$. The throughput curves are plotted for epsilon of 0.25, 0.5 and 0.75 for a fixed value of $\tau=0.25$. Similar to Figure \ref{fig:radarrange}, we observe that the throughput rises for each curve initially to reach a peak but then decreases with higher network density. By comparing the throughput curves for different $\epsilon$, we see that for the higher value of $\epsilon$ lower throughput can be reached. The reason is that since the time dedicated to radar operation increases, less time is dedicated to communication, leading to lower network throughput. To this end, we can infer that to serve throughput-hungry applications, JCAS nodes may assign lesser slots for radar operation.
\begin{figure}[t!]
\vspace{-0.15cm}
	\centering
	\includegraphics[width=0.43\textwidth,trim = 0mm 2.2mm 0mm 0.50mm,clip]{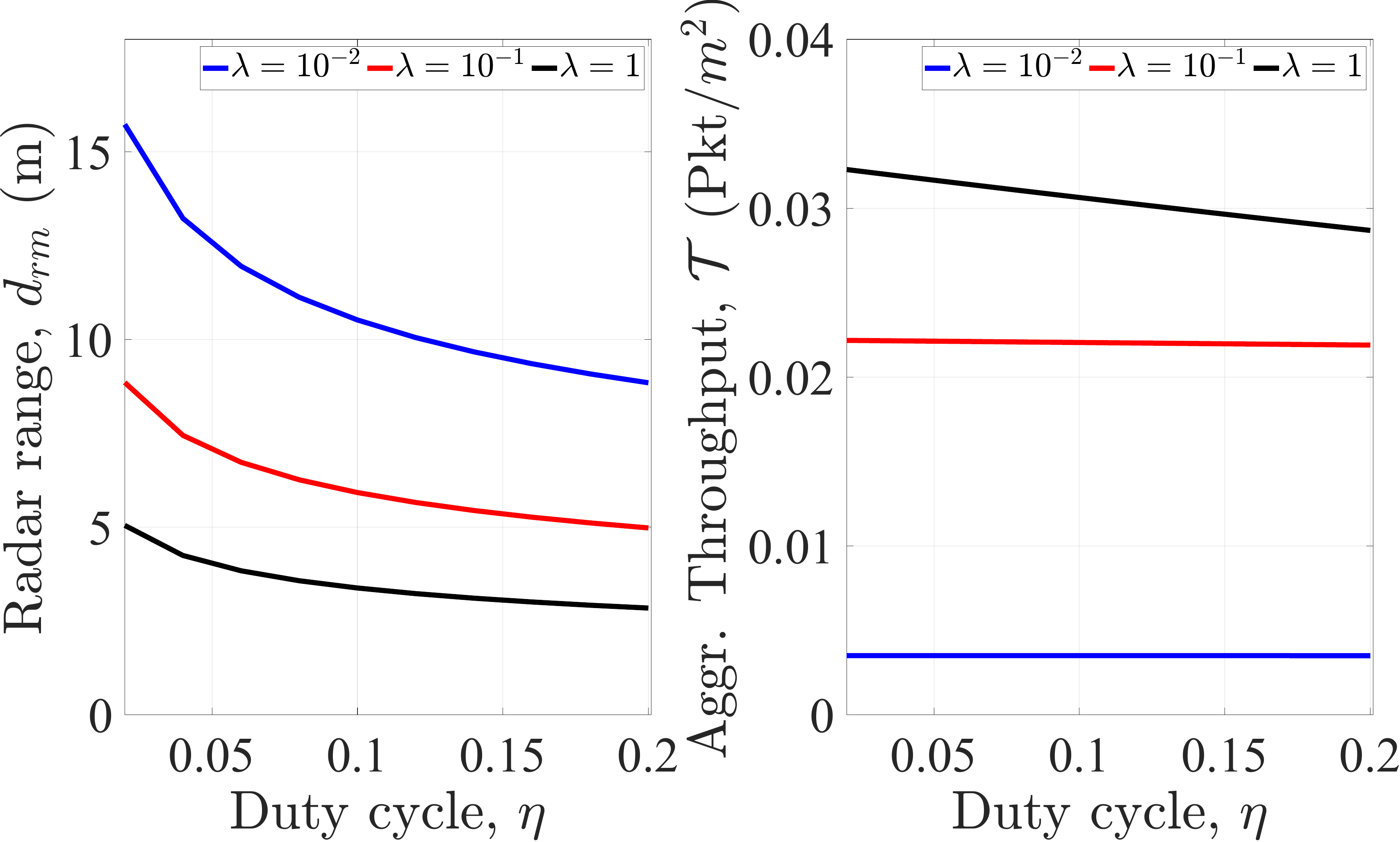}
	\caption{Effect of duty cycle }
	\label{fig:radarrange_th_dc}
	\vspace{-0.5cm}
\end{figure}

In Figure \ref{fig:radarrange_th_dc}, we show the relationship between the radar range and aggregated throughput with varying duty cycles. We show the results for $\lambda$ of $10^{-2}$, $10^{-1}$, and $1$. In the left subfigure, we observe that with the increasing radar duty cycle, the radar range decreases. This is because with a higher duty cycle, the radar-to-radar interference increases. In addition, each radar has a lesser waiting time for their echo \cite{tait2005}. Furthermore, with higher network density, radar performance in terms of range is also severely affected as discussed in Figure \ref{fig:radarrange}. In the right subfigure, we depict the effect of the duty cycle on the aggregated network throughput. It shows that for sparse networks ($\lambda=10^{-2}$), the effect of the duty cycle is negligible. However, for denser network scenario ($\lambda=1$), aggregated throughput decrease for higher duty cycle due to higher radar-to-communication interference. Finally, comparing both subfigures, we infer that there is a trade-off exists between the radar and communication performance in terms of duty cycle. By giving higher precedence to sensing, the communication performance decreases. Therefore, from a network management perspective, the JCAS network will achieve its greatest potential if the duty cycle of each node is tuned and managed in a suitable way to adapt and optimize the joint sensing and communication performance.

\section{Conclusions}
In this paper, a statistical geometry approach is used to model a CSMA-based JCAS network for sub-7 GHz with time-sharing sensing and radar functionalities. We verified the analytical model through simulations, and studied how different network parameters, affected both the communication and radar performance. Our analysis reveal key performance trade-offs between the radar detection range and the communication throughput.  While the radar detection range suffers from increased interference, the system was found to work well for short-range detection applications, even for dense network scenarios. Moreover, to serve throughput-intensive applications, JCAS nodes should assign less time for radar operation.

\small
\section*{Acknowledgment}
This work was partially funded by the Federal Ministry of Education and Research Germany within the project “Open6GHub” under grant 16KISK012 and partially funded under the Excellence Strategy of the Federal Government and the Länder. Simulations were performed with computing resources granted by RWTH Aachen University under project rwth0767.
\bibliographystyle{IEEEtran}
\bibliography{conference_101719}
\end{document}